\documentclass[aps,prl,reprint,amsmath,amssymb,superscriptaddress,floatfix,longbibliography]{revtex4-1}
\usepackage{amsthm,dsfont,ctable,url,braket,mathtools}
\usepackage{natbib}
\usepackage[colorlinks,
            linkcolor=blue,
            anchorcolor=blue,
            citecolor=blue,
urlcolor=blue
            ]{hyperref}
\usepackage{graphicx}
\usepackage[FIGTOPCAP,raggedright,nooneline]{subfigure}
\usepackage{float}
\usepackage{longtable}

\usepackage[T1]{fontenc}
\usepackage{lmodern}
\newtheorem{theorem}{Theorem}
\newtheorem{lemma}{Lemma}

\newcommand{\Tr}{\operatorname{Tr}}

\def\>{\ensuremath{\rangle}}
\def\<{\ensuremath{\langle}}
\newcommand*\diff{\mathop{}\!\mathrm{d}}

\begin{document}

\title{Entirety of Quantum Uncertainty and Its Experimental Verification}

\author{Jie Xie}
\altaffiliation{These authors contributed equally to this work.}
\affiliation{National Laboratory of Solid State Microstructures, College of Engineering and Applied Sciences and School of Physics, Nanjing University, Nanjing 210093, China}
\affiliation{Collaborative Innovation Center of Advanced Microstructures, Nanjing University, Nanjing 210093, China}

\author{Songtao Huang}
\altaffiliation{These authors contributed equally to this work.}
\affiliation{National Laboratory of Solid State Microstructures, College of Engineering and Applied Sciences and School of Physics, Nanjing University, Nanjing 210093, China}
\affiliation{Collaborative Innovation Center of Advanced Microstructures, Nanjing University, Nanjing 210093, China}

\author{Li Zhou}
\altaffiliation{These authors contributed equally to this work.}
\affiliation{Department of Computer Science and Technology, Tsinghua University, Beijing, China}

\author{Aonan Zhang}
\affiliation{National Laboratory of Solid State Microstructures, College of Engineering and Applied Sciences and School of Physics, Nanjing University, Nanjing 210093, China}
\affiliation{Collaborative Innovation Center of Advanced Microstructures, Nanjing University, Nanjing 210093, China}

\author{Huichao Xu}
\affiliation{National Laboratory of Solid State Microstructures, College of Engineering and Applied Sciences and School of Physics, Nanjing University, Nanjing 210093, China}
\affiliation{Collaborative Innovation Center of Advanced Microstructures, Nanjing University, Nanjing 210093, China}

\author{Man-Hong Yung}
\affiliation{Shenzhen Institute for Quantum Science and Engineering and Department of Physics, Southern University of Science and Technology, Shenzhen 518055, China}
\affiliation{Shenzhen Key Laboratory of Quantum Science and Engineering, Southern University of Science and Technology, Shenzhen 518055, China}

\author{Nengkun Yu}
\email[]{nengkunyu@gmail.com}
\affiliation{Centre for Quantum Software and Information, School of Software, Faculty of Engineering and Information Technology, University of Technology Sydney NSW}

\author{Lijian Zhang}
\email[]{lijian.zhang@nju.edu.cn}
\affiliation{National Laboratory of Solid State Microstructures, College of Engineering and Applied Sciences and School of Physics, Nanjing University, Nanjing 210093, China}
\affiliation{Collaborative Innovation Center of Advanced Microstructures, Nanjing University, Nanjing 210093, China}

\date{\today}

\begin{abstract}
As a foundation of modern physics, uncertainty relations describe an ultimate limit for the measurement uncertainty of incompatible observables. Traditionally, uncertain relations are formulated by mathematical bounds for a specific state. Here we present a method for geometrically characterizing uncertainty relations as an entire area of variances of the observables, ranging over all possible input states. We find that for the pair of position $x$ and momentum $p$ operators, Heisenberg's uncertainty principle points exactly to the area of the variances of $x$ and $p$. Moreover, for finite-dimensional systems, we prove that the corresponding area is necessarily semialgebraic; in other words, this set can be represented via finite polynomial equations and inequalities, or any finite union of such sets. In particular, we give the analytical characterization of the areas of variances of (a) a pair of one-qubit observables, (b) a pair of projective observables for arbitrary dimension, and give the first experimental observation of such areas in a photonic system.
 \end{abstract}

\maketitle

\par
\textit{Introduction.---}Quantum mechanics is revolutionizing our outlook on the world. The most dramatically changing may be ``god does play dice with the universe". In other words, quantum world is unpredictable inherently. Heisenberg's uncertainty principle is the most famous result of this unpredictability of the quantum world \cite{Heisenberg1927}. Roughly speaking, it asserts a fundamental limit to the precision with which the position and momentum can be known simultaneously. After that, many efforts have been devoted to understand this mystery. Recent years, the uncertainty relation still aroused a lot of research interests and shows more and more applications in quantum information science, such as providing separability criteria \cite{PhysRevLett.92.117903}, determining the nonlocality of quantum systems \cite{Oppenheim2010} and device-independent quantum cryptography \cite{PhysRevA.90.012332}. Many different types of uncertainty relations have been proposed \cite{RevModPhys.86.1261,RevModPhys.89.015002} and aroused a heated discussion, especially the error-disturbance relations on the joint measurement of two incompatible observables \cite{PhysRevA.67.042105,PhysRevLett.111.160405,PhysRevLett.112.050401,zhou2016verifying}.
 \par
 The most well-known uncertainty relation of position $x$ and momentum $p$ was rigorously proven by Kennard \cite{Kennard1927} and Weyl \cite{Weyl1927}, $\Delta x\Delta p \geq\frac{\hbar^2}{4}$, where $\Delta x$ and $\Delta p$ stand for the variances of $x$ and $p$. Then Robertson \cite{Robertson1929} extended it for two arbitrary observables $A$ and $B$, $\Delta A\Delta B\geq\frac{1}{4}|\langle[A,B]\rangle|^2$, where $\langle[A,B]\rangle$ stands for the ensemble average of the commutator $[A,B]$. Here the variances and average value are in terms of a particular state, and the above equation was later improved by Schr\"{o}dinger \cite{Schrodinger1930}. All the above uncertainty relations deal with the product of variances of quantum observables, which can be trivial when one of the variances goes to zero \cite{PhysRevLett.113.260401}, e.g., if the system is in an eigenstate of $A$, then $\Delta A=0$. Due to this problem, uncertainty relations in the form of summation of variances of observables has been proposed \cite{Robertson1929,Pati2007,PhysRevLett.113.260401}. Moreover, some researchers think entropy would be the more natural and reasonable way to characterize uncertainty \cite{Hirschman1957,Beckner1975,Deutsch1983,Maassen1988,PhysRevA.89.022112,Xiao2017}. However, these uncertainty relations do not give an overall description for all quantum states, in the sense that they either give a state-dependent \cite{Robertson1929,Pati2007,PhysRevLett.113.260401} bound or a lower bound that not optimal for some states\cite{Maassen1988,PhysRevA.89.022112}.

 \par
To give an universal description of the quantum uncertainty for all states, the concept of uncertainty region (UR) has been proposed, which is the set of variances $(\Delta O_1,\cdots, \Delta O_n)$ (here we mainly focus on the variance as uncertainty measure) of a set of observables $\{O_1, ..., O_n\}$ for all quantum states, naturally forming a geometric area in Euclidean space.
This concept was first proposed as a state-independent uncertainty relation in Ref. \cite{li2015reformulating}, which gave the variance-based uncertainty region for a pair of traceless qubit observables. Then Abbott et al. \cite{Abbott2016} generalized it to the UR of any number of $\pm1$-valued Pauli observables (both standard derivation-based and entropy-based). Recently, nontrivial URs of some specific pairs of qutrit observables was also given \cite{Busch2019}, while URs of a pair of arbitrary qutrit observables or higher dimensional nontrivial quantum observables is still very difficult to derive, because the quantum state space becomes very complicated as the dimension grows.

In this work, we first propose the concept of UR for pure states and mixed states and show that both uncertainty regions are semialgebraic sets in a Euclidean space. We characterize the URs of a pair of projective operators, which shows simpler results than those of Pauli observables and can be easily extended to arbitrary quantum observables
in the qubit case. For infinite-dimensional system, we show that the standard Heisenberg's uncertainty principle represents the exact UR of $x$ and $p$. For finite-dimensional system, we give the analytical characterization of URs of a pair of arbitrary qubit and qudit projectors. In addition, we give the first experimental observation of quantum URs for qubit and qutrit projective measurements with photonic setups. Our experimental results show good agreement with to the theoretical prediction.

\textit{Theory.---}For a set of observables $\{O_1,\cdots,O_n\}$, finite or infinite-dimensional, we are interested in the following two URs, $R_p:=\{(\Delta O_{1,\psi},\cdots,\Delta O_{n,\psi}):\<\psi|\psi\>=1\}$ and $R_m:=\{(\Delta O_{1,\rho},\cdots,\Delta O_{n,\rho}):\rho\geq0,\Tr\rho=1\}$, namely the URs for pure states and mixed states, where $\Delta O_{i,\psi}=\bra{\psi}O_i^2\ket{\psi}-\bra{\psi}O_i\ket{\psi}^2$ and $\Delta O_{i,\rho}=\Tr(O_i^2\rho)-\Tr^2(O_i\rho)$, denoting the variance of $O_i$. It is obvious that $R_p\subset R_m$, besides this, we have the following characterization.
\begin{theorem}
In finite-dimensional Hilbert space, $R_m$ and $R_p$ are both semialgebraic sets. A semialgebraic set in $\mathbb{R}^n$ is defined by the finite union of the sets determined by a finite number of polynomial equations $p(x_{1},...,x_{n})=0$ and inequalities $q(x_{1},...,x_{n})>0$.
\end{theorem}
\begin{proof}We can parameterize $\rho$ using $d^2-1$ real parameters $x_1,\cdots,x_{d^2-1}$, then observe that
\begin{align*}
S_m:=\{(\Delta O_{1,\rho},\cdots,\Delta O_{n,\rho},x_1,\cdots,x_{d^2-1}): \\
\rho\sim(x_1,\cdots,x_{d^2-1})\}
\end{align*}
is semialgebraic. One of the most important property of semialgebraic sets, known as Tarski-Seidenberg theorem~\cite{coste2000introduction}, is that they are closed under the projection operators: that is, projecting a semialgebraic set onto a linear subspace yields another semialgebraic set. Therefore, $R_m$ is semialgebraic. The same idea shows that $R_p$ is semialgebraic where only $2d-1$ real parameters are used.
\end{proof}
Besides, when considering the UR of two projective operators, we have the following conclusion.
\begin{theorem}
For two projections $P$ and $Q$, we can always have
$$R_p=R_m.$$
\end{theorem}However, this is not true generally and we leave the proof of this theorem in Supplemental Material \cite{Supplementary}.

For infinite-dimensional system, we are interested in the joint distribution of $\Delta x$ and $\Delta p$. Generally, the physical system might be complex and it is not easy to characterize the exact uncertainty region. Here, we only consider the wave packet for free particles.
\begin{theorem}
For one-dimensional free particles, two uncertainty regions for pure states and mixed states are the same:
\begin{align*}
R_p=R_m=\{(x,y):xy\geq \frac{\hbar}{2},x,y>0\}.
\end{align*}
\end{theorem}
The proof of this theorem is completed by considering the one-dimensional Gaussian wave packet. For the tedious calculations, we also leave it in the Supplemental Material \cite{Supplementary}.

Next, we start to give the analytical characterization of the areas of UR for a pair of projective operators. Before that, we introduce the following famous Jordan's lemma \cite{Scarani2019} used in the characterization.
\begin{lemma}
\label{lemma:1}
Let $P$ and $Q$ be two projection operators in $\mathbb{C}^d$. Then there exists a basis of $\mathbb{C}^d$ in which $P$ and $Q$ are simultaneously block diagonal, with $l$ blocks of size one or two such that either (for one-dimensional blocks)
\begin{align*}
P_i,Q_i\in\{(0),(1)\},
\end{align*}
or (for two-dimensional blocks)
\begin{align}
\label{eqn:P_i_Q_i}
P_i=\begin{pmatrix}1   & 0\\0  &  0\end{pmatrix},Q_i=\begin{pmatrix}\cos^2\theta_i & \cos\theta_i\sin\theta_i\\ \cos\theta_i\sin\theta_i & \sin^2\theta_i\end{pmatrix}
\end{align}
with $\theta_1,\cdots,\theta_l\in(0,\frac{\pi}{2}]$.
\end{lemma}

For qubit system, we only need to care about rank 1 projectors. This is because for any Hermitian observable $O_i$, there exists a $\lambda_i\in\mathbb{R}$ such that $O_i-\lambda_i I$ is rank 1 and $R_m$ and $R_p$ for a pair of qubit projectors can be easily obtained, as shown in the following. Thus by using this lemma and the qubit case, we can compute $R_m$ and $R_p$ of two projective observables for higher dimensional systems. For the convenience of the following clarification, we denote the $d$-dimensional URs of two projective observables as $R^{(d)}$, and $R^{(d)}_p$ or $R^{(d)}_m$ for pure or mixed states, respectively.

Without loss of generality, we consider two trace one qubit projectors $A$ and $B$ written in the Bloch representation that $A=\frac{1}{2}(I + \vec{a}\cdot\vec{\sigma}),B=\frac{1}{2}(I + \vec{b}\cdot\vec{\sigma})$, where $\vec{a}=(a_x,a_y,a_z),\vec{b}=(b_x,b_y,b_z)$ are qubit Bloch vectors and $\vec{\sigma}=\sigma_x\vec{i}+\sigma_y\vec{j}+\sigma_z\vec{k}$ is the Pauli operator vector. Since $R_m=R_p$, we only considering an arbitrary density matrix $\rho=\frac{1}{2}(I + \vec{r}\cdot\vec{\sigma})$ with Bloch vector $\vec{r}=(r_x,r_y,r_z)$. According to Lemma \ref{lemma:1}, $\vec{a}$ and $\vec{b}$ can always be written in a basis such that $\vec{a}=(0,0,1)$ and $\vec{b}=(\sin2\theta,0,\cos2\theta)$, so that $A,B$ have the exact form of $P_i,Q_i$ in Eq. (\ref{eqn:P_i_Q_i}) and we define $\theta$ as the angle between $A$ and $B$. For two known projectors, the only constraint that restricts the obtainable region of $(\Delta A,\Delta B)$ is that $r^2\leq1$. Thus by this constraint, $R^{(2)}$ can be derived as the union set of the following two regions (see the details of derivation in Supplemental Material \cite{Supplementary})
\begin{align}
              \label{eqn:R_1}
	R_1^{(2)}:
	\left\{
	\begin{array}{lr}
	\Delta A+\Delta B\geq (1+\cos^22\theta)/4,&\\
	\Delta A,\Delta B\leq 1/4,
	\end{array}
	\right.
	\end{align}
	\begin{align}
       \label{eqn:R_2}
	R_2^{(2)}:
	\left\{
	\begin{array}{lr}
	\Delta A+\Delta B\leq(1+\cos^22\theta)/4,&\\
	64x_1^2/(1+\cos4\theta)+64y_1^2/(1-\cos4\theta)\leq1,
	\end{array}
	\right.
	\end{align}
where $x_1$ and $y_1$ are
	\begin{align*}
		\left( \begin{array}{ccc}
		x_1 \\
		y_1
		\end{array}
		\right)=
		\left(
		\begin{array}{ccc}
		\cos{\frac{\pi}{4}} & \sin{\frac{\pi}{4}} \\
		-\sin{\frac{\pi}{4}} & \cos{\frac{\pi}{4}}
		\end{array}
		\right)
		\left(
		\begin{array}{ccc}
		\Delta A-{1}/{8} \\
		\Delta B-{1}/{8}
		\end{array}
		\right).
	\end{align*}
It is obviously that $R_1^{(2)}$ represents a right triangle in the $\Delta A-\Delta B$ plane while $R_2^{(2)}$ is part of an ellipse with one of its axis tilted at $\theta=\frac{\pi}{4}$ against the $\Delta A$ coordinate axis and centered at $({1}/{8},{1}/{8})$, as show in Fig. \ref{fig:qubit_th:a}.
 The UR $R^{(2)}$ is composed of the union of $R_1^{(2)}$ and $R_2^{(2)}$. It also can be seen that points on the boundary of this ellipse can only be achieved by states on the equatorial plane of the Bloch ball satisfying $r_x^2+r_z^2=1$. We call these states as "boundary states", which generate most of boundary points of $R^{(2)}$. As $\theta$ grows larger, the long axis of ellipse shrinks while the minor axis gets longer. When $\theta=\pi/4$, the original long axis become 0 and the ellipse become a line, as depicted in Fig. \ref{fig:qubit_th:a}. As for the point $(1/4,1/4)$ in $R_1^{(2)}$, it can be achieved by the maximum mixed state or pure state with Bloch vector $\vec{r}$ perpendicular to $\vec{a}$ and $\vec{b}$, indicating that $A$ and $B$ achieve maximum uncertainty simultaneously. For other points in $R_1^{(2)}$, it is also achievable for pure states.

As mentioned above, in the qubit case, this characterization is not limited to projective observables. The UR of two Hermitian observable $O_1,O_2$ can be characterized by transform them to two projective observables through $O_1-\lambda_1I$ and $O_2-\lambda_2I$, which share the same UR.
\begin{figure*}
  \centering
  \addtocounter{figure}{1}
  \subfigure[]{
    \label{fig:qubit_th:a} 
    \begin{minipage}[b]{1\textwidth}
    \includegraphics[width=0.23\textwidth]{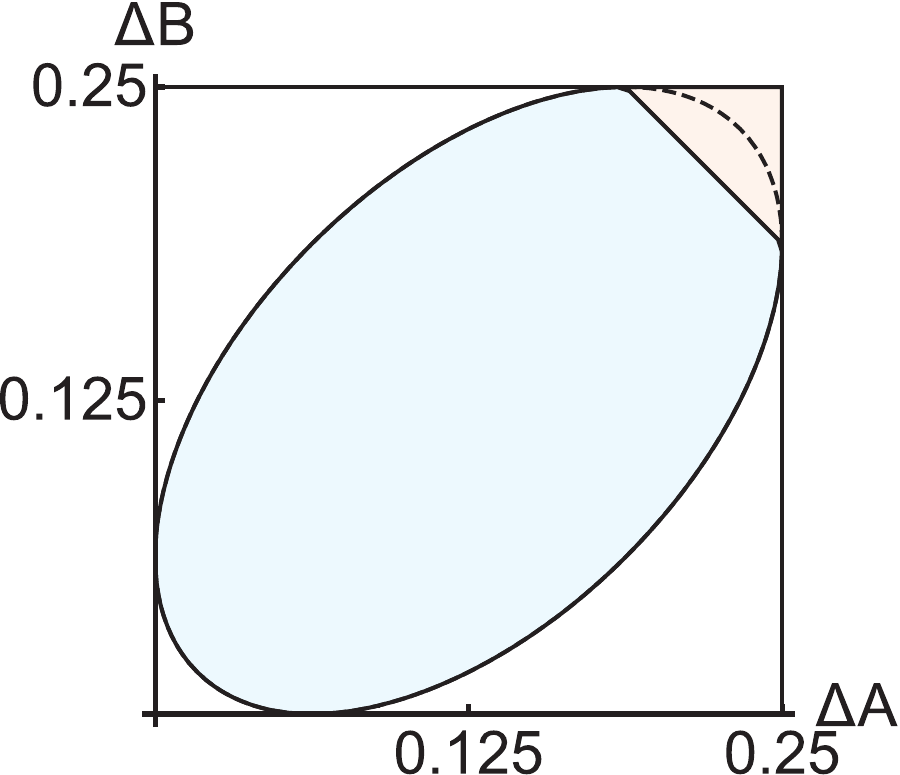}
    \includegraphics[width=0.23\textwidth]{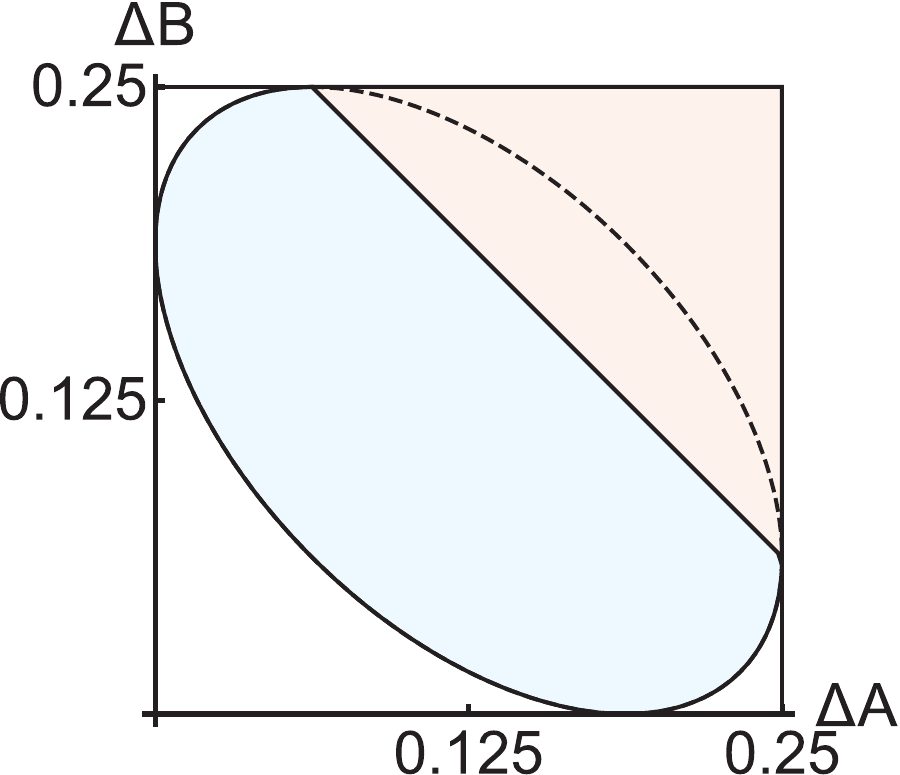}
    \includegraphics[width=0.23\textwidth]{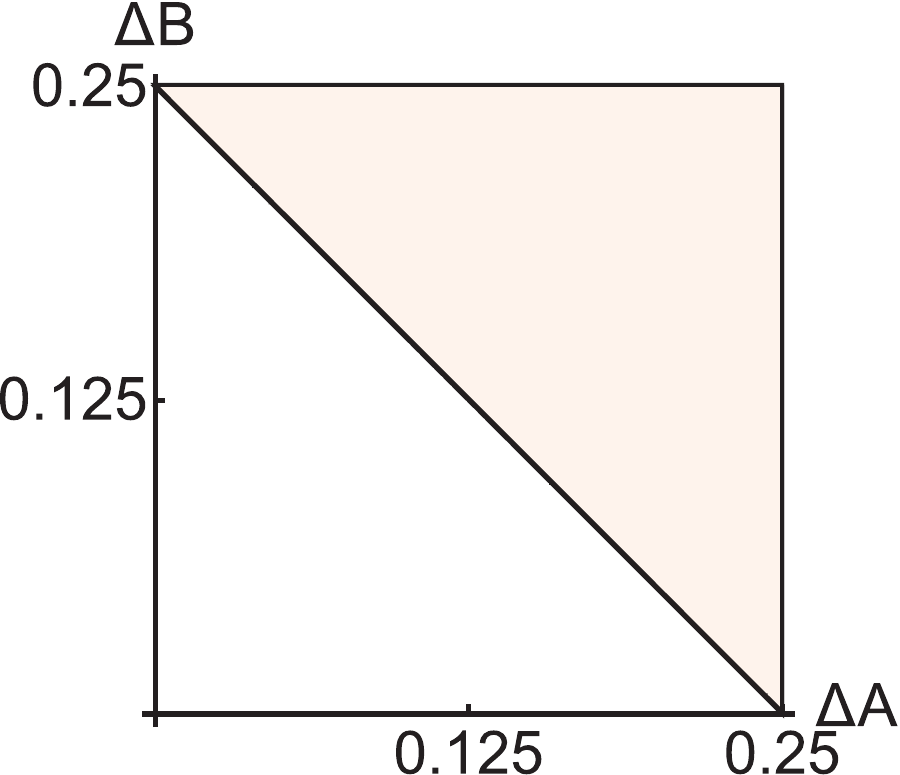}
    \includegraphics[width=0.23\textwidth]{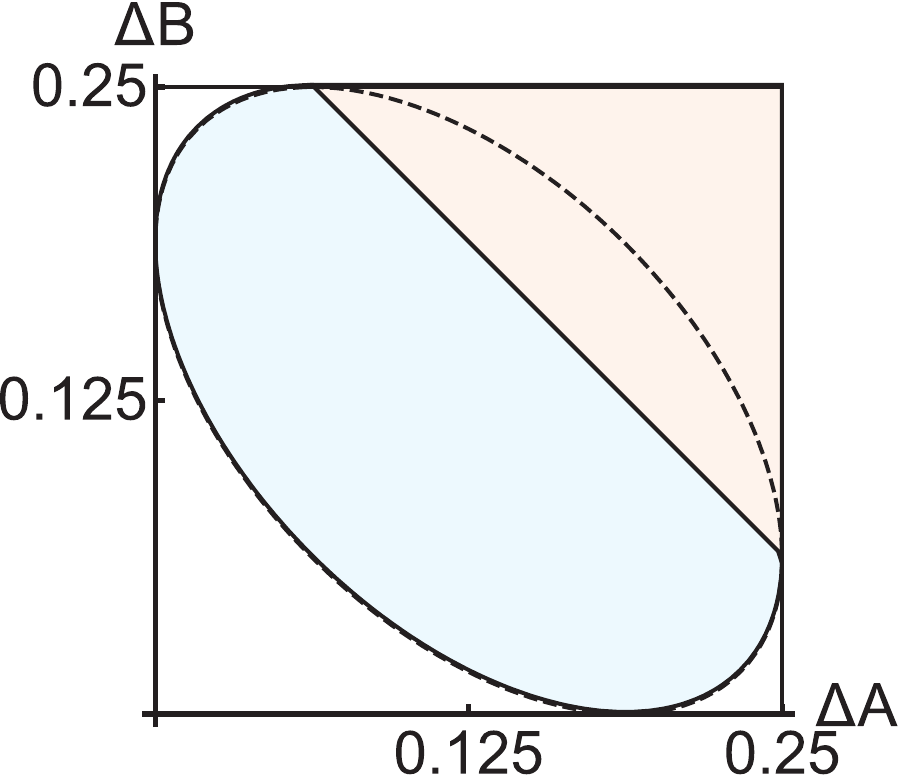}
    \end{minipage}
    }
  \subfigure[]{
    \label{fig:qutrit_th:b} 
    \begin{minipage}[b]{1\textwidth}
    \includegraphics[width=0.23\textwidth]{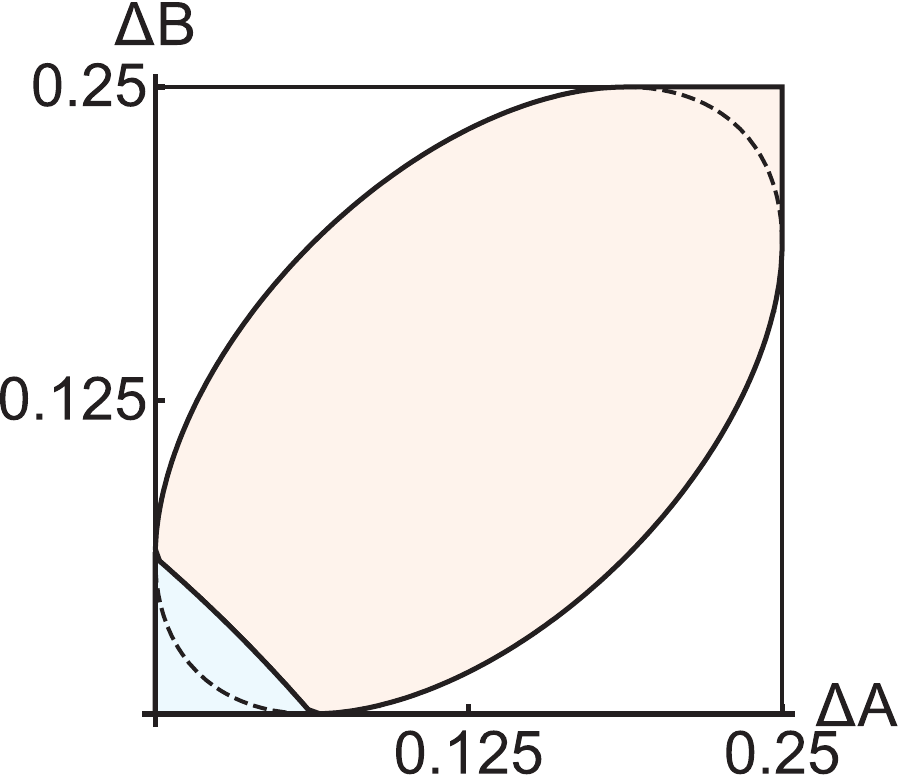}
    \includegraphics[width=0.23\textwidth]{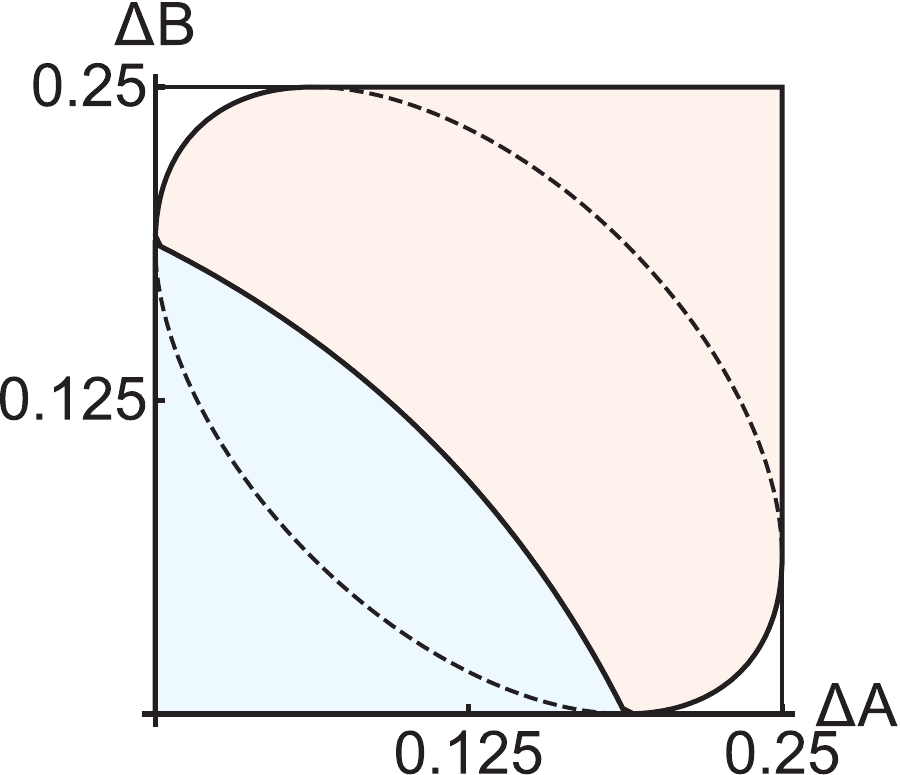}
    \includegraphics[width=0.23\textwidth]{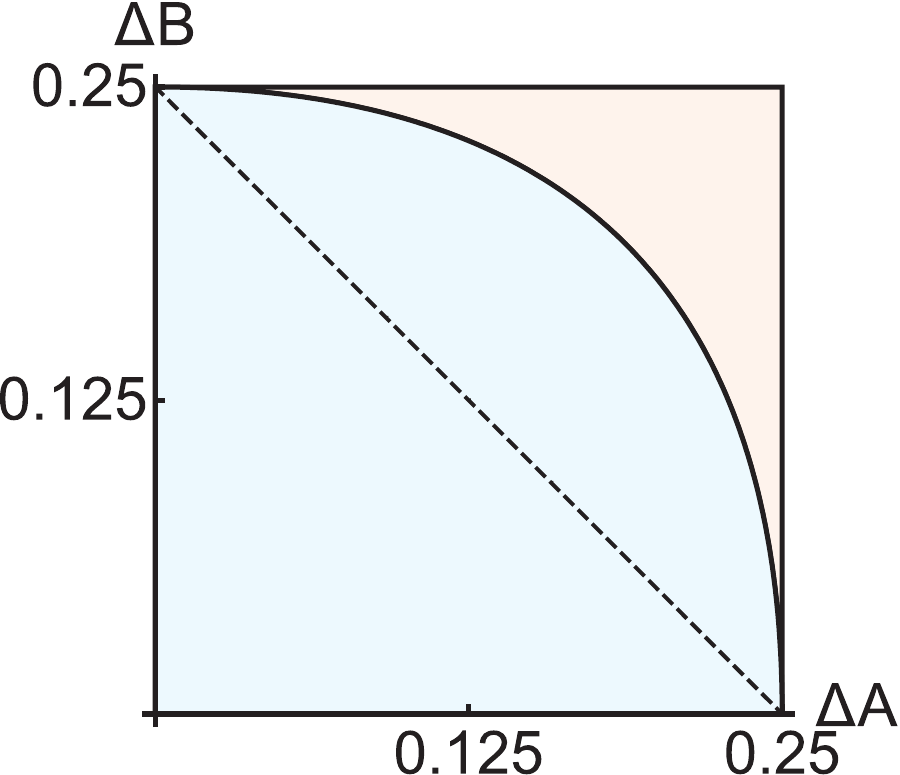}
    \includegraphics[width=0.23\textwidth]{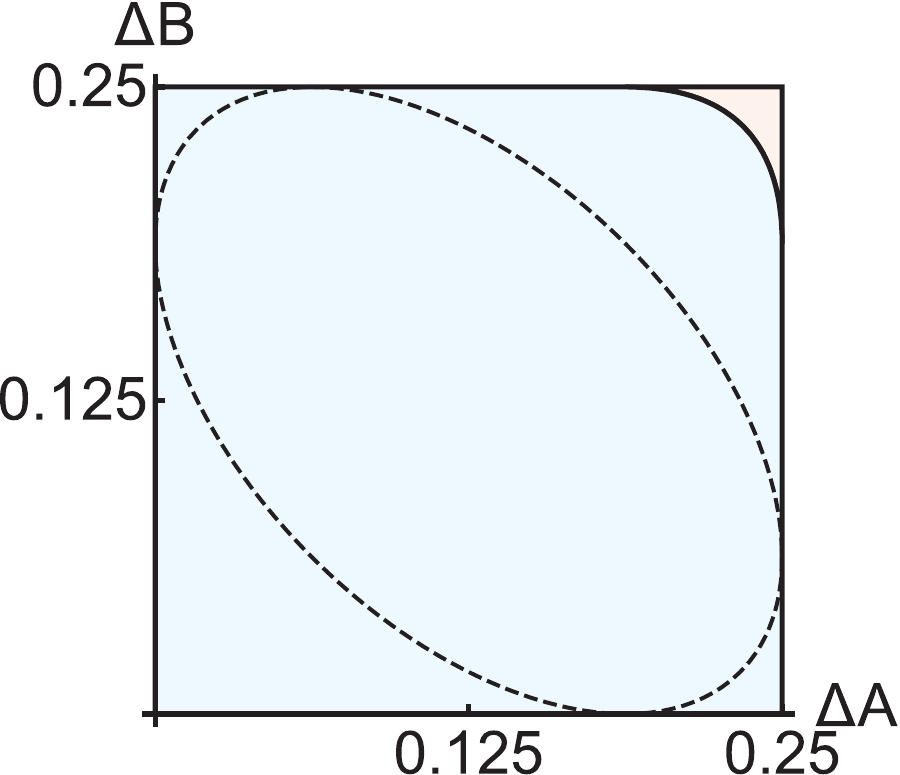}
    \end{minipage}
    }
  \addtocounter{figure}{-1}
  \caption{Theoretical uncertainty regions of a pair of (a) qubit and (b) qutrit projective observables with different included angles $\theta={\pi}/{12},{\pi}/{6},{\pi}/{4},{\pi}/{3}$ (from left to right). $R_1$ and $R_2$ are filled with light yellow and light blue colors, respectively.}
  \label{fig:results_theory} 
\end{figure*}

With the above results of $R^{(2)}$ for qubit projectors and applying Lemma \ref{lemma:1}, the problem of finding $R^{(d)}(d\geq3)$ can be reduced into qubit subspace of higher dimensional systems. More specifically, we can always find a basis in which arbitrary two rank 1 projective operators can be written as
	\begin{align}
    \label{A_B_d}
	A=\begin{pmatrix}
	1 & 0 & \quad \\
	0 & 0 & \quad \\
	\quad & \quad & \ddots
	\end{pmatrix},
    B=\begin{pmatrix}\cos^2\theta & \cos\theta\sin\theta & \quad \\ \cos\theta\sin\theta & \sin^2\theta & \quad \\ \quad & \quad & \ddots
	\end{pmatrix},
	\end{align}
where all the blank space equal to 0, which indicate that the two observables are reduced to a qubit subspace. Thus the density matrix $\rho$ can also be equivalently written in the qubit subspace, by setting all the elements equal to 0 apart from the $2\times2$ block in the left upper corner, as follows:
	\begin{align}
    \label{rho_d}
	\rho = \begin{pmatrix} \rho^{(2)} \\
	& & \ddots
	\end{pmatrix}.
	\end{align}
Here $\rho^{(2)}$ is the nonzero $2\times2$ block representing an unnormalized qubit density matrix, for $\rho^{(2)}$ no longer satisfies the trace one requirement but still need to be positive. Therefore, we can write $\rho^{(2)}$ as
	\begin{equation*}
      \label{rho_2}
	  \rho^{(2)}=\frac{1}{2}(\alpha I+\vec{r}\cdot\vec{\sigma}),r^2=1
	\end{equation*}
with $0\leq\alpha\leq1$, and the positivity of $\rho^{(2)}$ suggests $r^2\leq\alpha^2$. Here $\alpha=\Tr (\rho^{(2)})$ is the probability that $\rho$ been projected onto the qubit subspace. From the above statements in Eq. (\ref{A_B_d}) and Eq. (\ref{rho_d}), it is obvious that all the calculations can be reduced in the qubit subspace. The only difference is that the density matrix $\rho^{(2)}$ become an unnormalized one. Similar to the qubit case, by using the constraint $r^2\leq\alpha^2$, we can derive the following two regions that form $R^{(d)}$ (see Supplemental Material for more details \cite{Supplementary}),
\begin{align}
              \label{eqn:R_1_d3}
	R_1^{(d)}:
	\left\{
	\begin{array}{lr}
	2-\sqrt{1-4\Delta A}-\sqrt{1-4\Delta B}\geq2\sin^2{\theta},&\\
    \quad  \\
    \pm2\cos2\theta\sqrt{(1-4\Delta A)(1-4\Delta B)} \\
    -4(\Delta A+\Delta B)+(1+\cos^22\theta)\leq0,
	\end{array}
	\right.
	\end{align}
	\begin{align}
       \label{eqn:R_2_d3}
	R_2^{(d)}: 0\leq2-\sqrt{1-4\Delta A}-\sqrt{1-4\Delta B}<2\sin^2{\theta}.
\end{align}
The second inequality in Eq. (\ref{eqn:R_1_d3}) gives the exact region of $R^{(2)}$ and corresponds to the situation where $\alpha=1$ in qutrit and higher dimensional cases. The first inequality in Eq. (\ref{eqn:R_1_d3}) cuts part of the ellipse in $R^{(2)}$ and forms the concave curve in $R_1^{(d)}$, which is part of a parabolic curve, as shown in Fig. \ref{fig:qutrit_th:b}. While $R_2^{(d)}$ gives the bottom left corner, which can be derived from states with $\alpha<1$. $R_1^{(d)}$ and $R_2^{(d)}$ together forms the whole uncertainty region of $R^{(d)}$. When $\alpha=0$, the region is reduced to the original point $(0,0)$, which is intuitive since this state can never be projected into the qubit subspace, and the measurement outcomes of $A,B$ are simultaneously determined. Compared to the qubit case, due to the unnormalization of $\rho^{(2)}$, the mainly difference of higher dimensional system is that the bottom left corner can also be obtained. $R_1^{(d)}$ behaves as the same way as in the qubit case when $\theta$ changes while $R_2^{(d)}$ become larger as $\theta$ increase. When $\theta=\pi/4$, the lower left area is fully filled and when $\theta>\pi/4$, $R_2^{(d)}$ keeps getting larger, thus the UR covers the whole box. This is another property that qubits do not share.

\begin{figure}[t]
\centering
\includegraphics[width=1\linewidth]{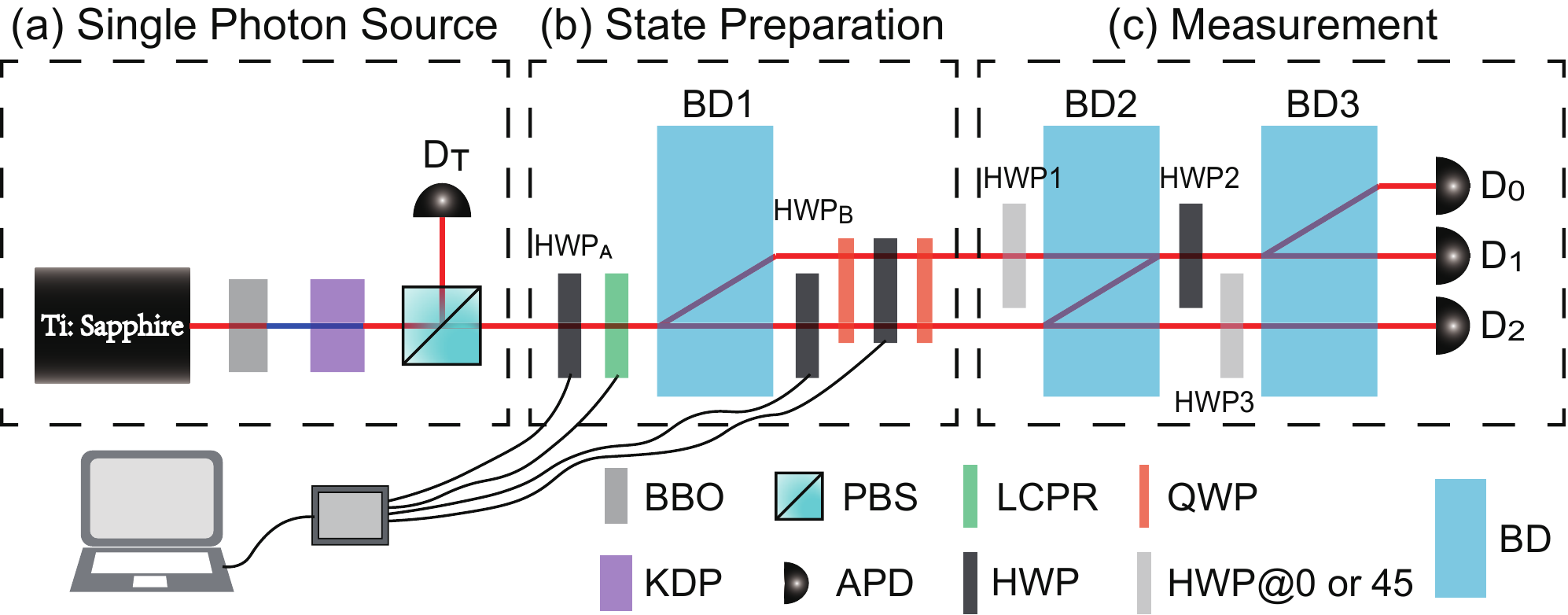}
\caption{Schematic picture of experimental set-up. (a) Frequency-doubled Ti:Sapphire laser pulses (~150fs duration, 415nm) pump a phase-matched potassium di-hydrogen phosphate (KDP) crystal to generate photon pairs. The photon pairs are then divided into signal and idler modes by a polarizing beamsplitter (PBS). Detection of one photon in the idler mode heralds a single photon in the signal mode. (b) The state preparation module is composed of two electronically-controlled half-wave plates (HWPs) ,a liquid crystal phase retarder (LCPR), a birefringent calcite beam displacer (BD) and another electronically-controlled HWP inserted in two quarter-wave plates (QWPs), which include a relative phase between horizontal and vertical polarizations. (c) Three HWPs and two BDs project the single photon into three orthogonal states, then three avalanche photodiodes (APDs) are used for detection of the single photon.}
\label{fig:exp_setup}
\end{figure}

\begin{figure*}[t]
  \centering
  \addtocounter{figure}{1}
  \subfigure[]{
    \label{fig:qubit:b} 
    \begin{minipage}[b]{1\textwidth}
    \includegraphics[width=0.23\textwidth]{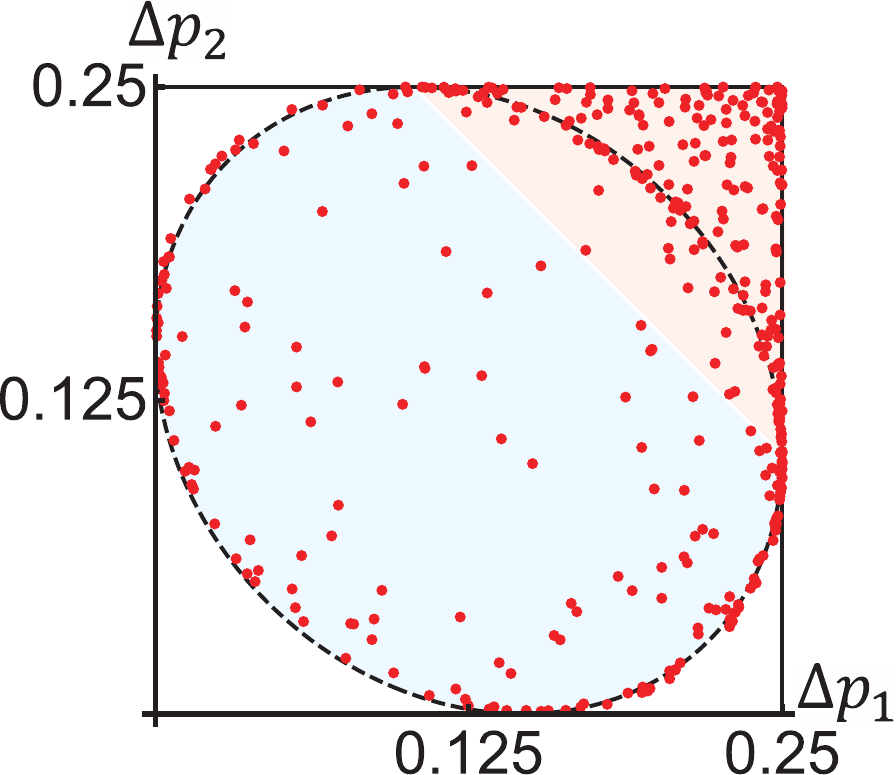}
    \includegraphics[width=0.23\textwidth]{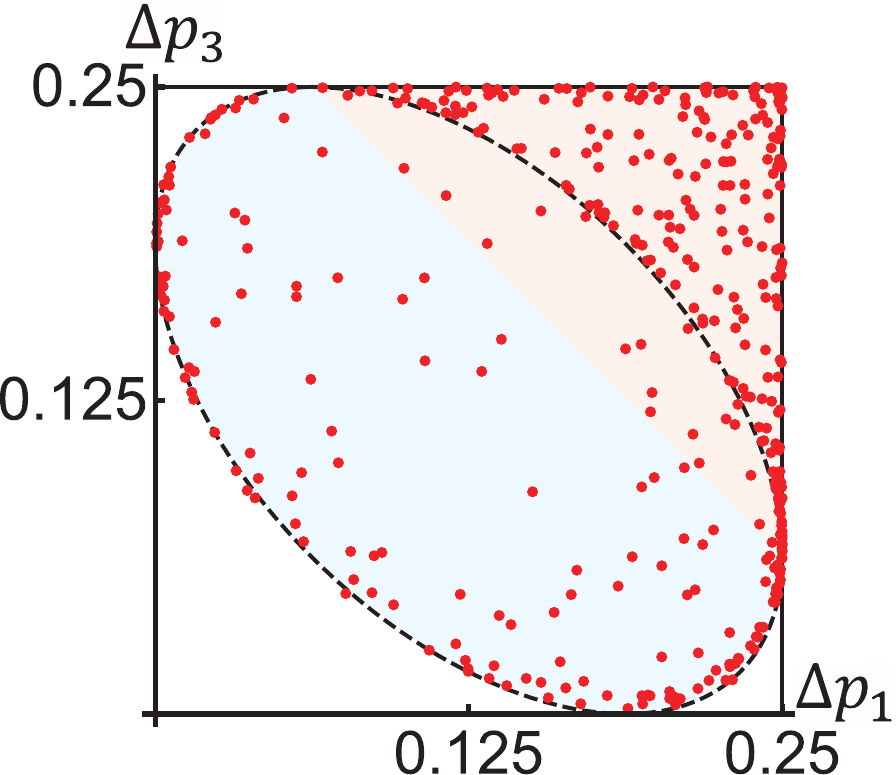}
    \includegraphics[width=0.23\textwidth]{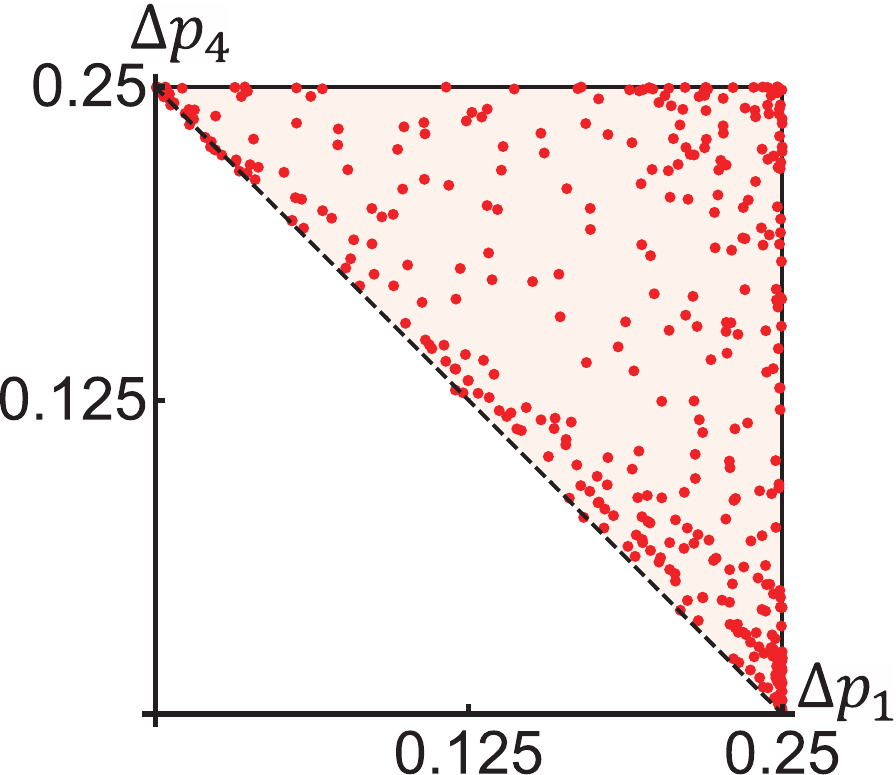}
    \includegraphics[width=0.23\textwidth]{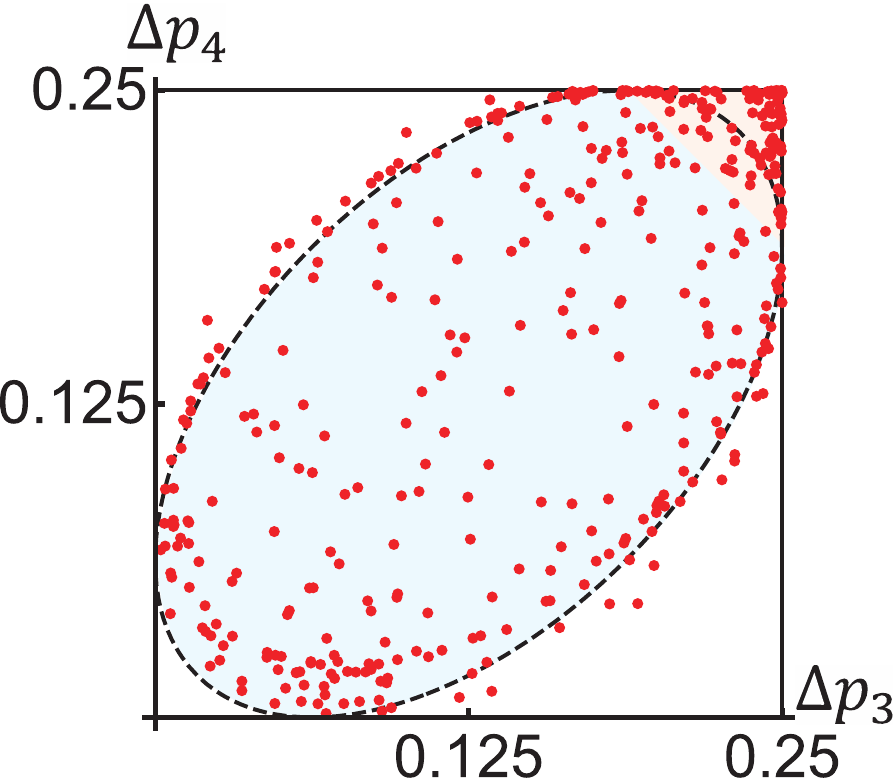}
    \end{minipage}
    }
  \subfigure[]{
    \label{fig:qutrit:a} 
    \begin{minipage}[b]{1\textwidth}
    \includegraphics[width=0.23\textwidth]{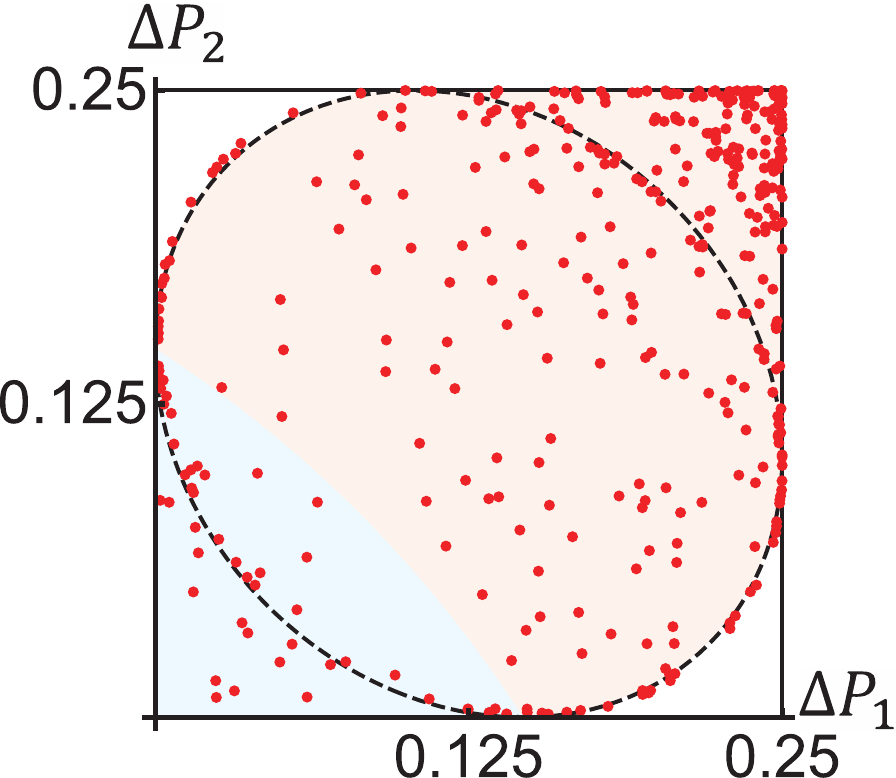}
    \includegraphics[width=0.23\textwidth]{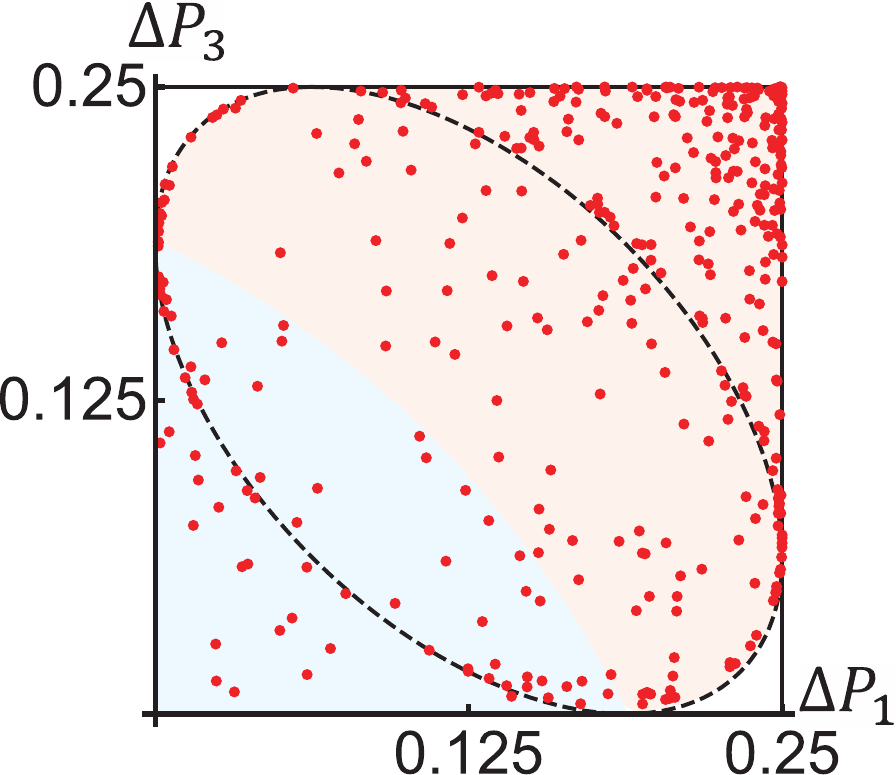}
    \includegraphics[width=0.23\textwidth]{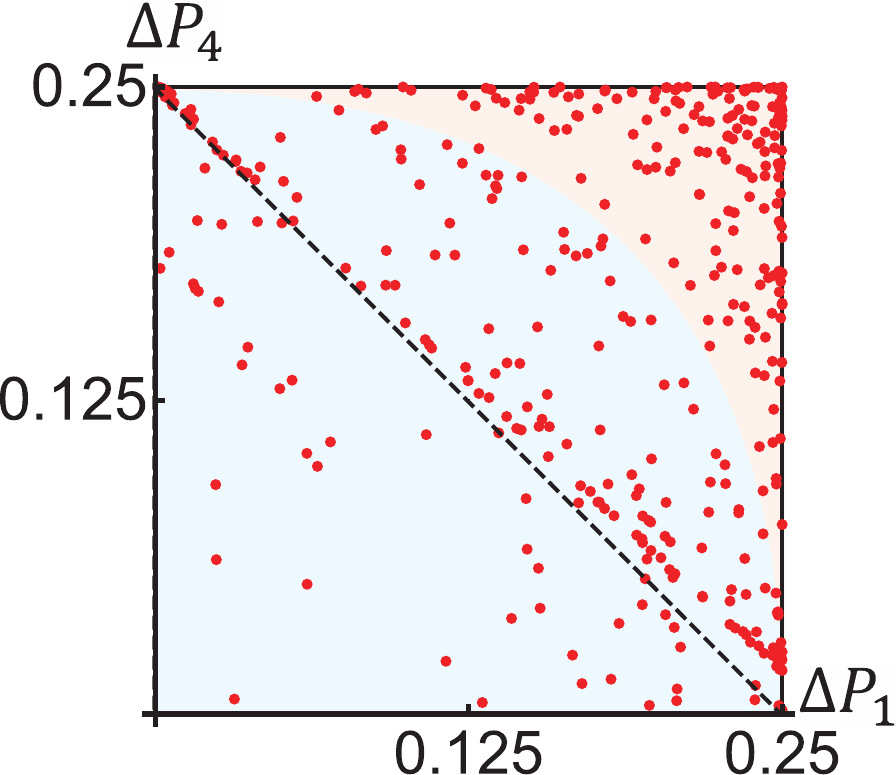}
    \includegraphics[width=0.23\textwidth]{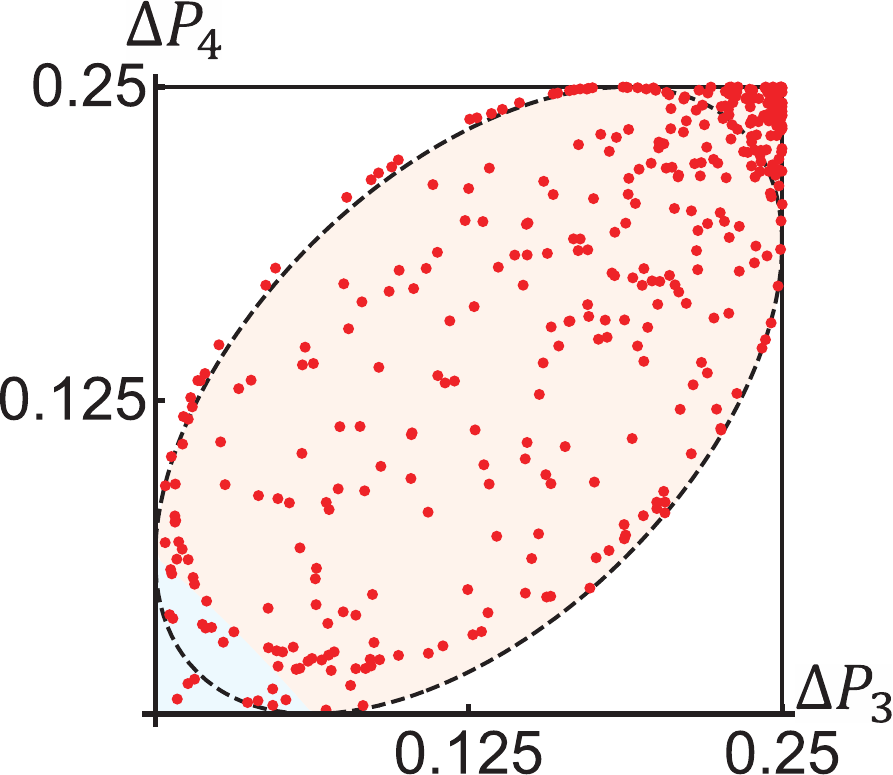}
    \end{minipage}
    }
  \addtocounter{figure}{-1}
  \caption{Experimentally observed uncertainty regions (URs) of (a) qubit and (b) qutrit with $\theta_{i,j}={5\pi}/{36},{\pi}/{6},{\pi}/{4},{\pi}/{12}$, respectively. $R_1$ and $R_2$ of theoretical URs are filled with light orange and light blue colors, respectively. The boundary of the ellipse is plotted with black dashed lines while the red points represent the experimentally observed data of 400 sampled states.}
  \label{fig:results_exp} 
\end{figure*}

\par
\textit{Experiment.---}We experimentally observe the qubit and qutrit uncertainty regions of different pairwise projective observables with photonic setups, as depicted in Fig. \ref{fig:exp_setup}. This set-up is specially designed to make the measurement of qutrit uncertainty regions compatible with the measurement of qubit ones, which can be implemented by blocking the third output port or post-selecting the qutrit measurement outcomes. When operating on the qutrit situation, single photons generated from spontaneous parametric down conversion firstly been scrambled in the superposition of horizontal and vertical polarizations by the half-wave plate (HWP) $\rm{WP_A}$, then a relative phase is introduced by the liquid crystal phase retarder (LCPR). After a birefringent calcite beam displacer (BD) and the HWP $\rm{WP_B}$, single photons are encoded in the superposition of polarization and spatial optical modes. The another phase retarder is realized by two quarter-wave plates (QWP) and a HWP, which constitute a QWP-HWP-QWP configuration and add a relative phase between horizontal and vertical polarizations. After that, an arbitrary pure qutrit state is generated (see more details in Supplemental Material \cite{Supplementary}). For example, in our experiment, we define the three eigen-modes of a qutrit state as
\begin{equation*}\label{polar_spatial_to_compu}
  \ket{0}\equiv\ket{H}\otimes\ket{s_1}, \ket{1}\equiv\ket{H}\otimes\ket{s_2},
  \ket{2}\equiv\ket{V}\otimes\ket{s_2},
\end{equation*}
where $\ket{s_1}, \ket{s_2}$ denote the spatial mode 1 (top) and spatial mode 2 (lower), $\ket{H}, \ket{V}$ denote the horizontal and vertical polarization, respectively. By changing the setting angles of $\rm{WP_A}$ and $\rm{WP_B}$ and applying appropriate parameters to the two phase retarder, different pure qutrit states can be prepared.

In the measurement part, we use two BDs and three HWPs to implement a $3\times3$ unitary transformation on the input state, then three avalanche photodiodes (APDs) are used for detection of the single photon, which as a whole equivalently performing three orthogonal projective measurements on the input state \cite{PhysRevLett.73.58}. By setting the angle of $\rm{HWP_3}$ to $0^\circ$ and the angle of $\rm{HWP_2}$ to $\theta_2$, the unitary transformation acts only on the first two eigen-modes and leave the third mode unchanged. Thus by changing the setting angle $\theta_2$ of $\rm{HWP_2}$, different projective measurements can be realized in the output ports $\rm{D_0}$ and $\rm{D_1}$. For example, $\rm{D_0}$ port corresponds to the projector (see more details in Supplemental Material \cite{Supplementary})
\begin{equation}\label{eqn:P}
  P=\begin{pmatrix}
  \cos^2 2\theta_{2} & \sin 2\theta_{2}\cos 2\theta_{2} & 0\\
  \sin 2\theta_{2}\cos 2\theta_{2} & \sin^2 2\theta_{2} & 0\\
  0 & 0 & 0\\
  \end{pmatrix}.
\end{equation}
Therefore, by inputting the same state into the measurement devices for enough trials, we can obtain the variances of different projectors for this input state, and consequently a point in the UR. By sampling enough states, the whole qutrit UR of a pair of projective measurements can be observed. In this configuration, we can also realize the observation of qubit URs by simply block detector $\rm{D_2}$ or post-selecting the detection events of detector $\rm{D_0}$ and $\rm{D_1}$.


\par

In our experiment, we randomly sample two families of pure qutrit states. The first is 300 arbitrary pure qutrit states, and the other is 100 "boundary states" with eigen-mode $\ket{2}$ fully unoccupied and a fixed phase ($0$ or $\pi$) between mode $\ket{0}$ and $\ket{1}$, which lies on the boundary of the ellipse in Eq. (\ref{eqn:R_2}) and Eq. (\ref{eqn:R_1_d3}). By setting $\theta_2=0,{5\pi}/{72},{\pi}/{12},{\pi}/{8}$, four projective measurements on the total 400 states is realized in detector $\rm{D_0}$, with an angle $2\theta_2$ between the eigen-projector $\ket{0}\bra{0}$. We denote the four projective measurement operators as $P_1$, $P_2$, $P_3$ and $P_4$, respectively.
Together with detector $\rm{D_1}$, detector $\rm{D_2}$ and the heralding detector $\rm{D_T}$, the statistics of measurement outcomes are registered by a coincidence logic with a coincidence window of 4.5ns. For each input state, the average registered photon counting events is about 45000 during one experiment and the experiment was repeated 5 times. From the measured statistics, the average values and variances of the four projectors for 400 input states were calculated. Dividing the four projectors into four pairs of observables, $\{P_1,P_2\}$, $\{P_1,P_3\}$, $\{P_1,P_4\}$ and $\{P_3,P_4\}$, their URs are shown in Fig. \ref{fig:qutrit:a}.  By post-selecting and normalizing the measured statistics in $\rm{D_0}$ and $\rm{D_1}$, we also derive the expectation values and variances of four qubit projectors, which is the non-zero two-dimensional block in $P_j$, denoted as $p_j(j=1,2,3,4)$. Their URs are shown in Fig. \ref{fig:qubit:b}.

For each pair of projectors $P_j$($p_j$) and $P_k$($p_k$), written in the form of Eq. (\ref{eqn:P}), their corresponding angle $\theta_{i,j}$ can be derived as $\theta_{j,k}=2(\theta_{2k}-\theta_{2j})$. It can be seen that both the qutrit and qubit experimental results agree well with the theoretical predictions for different values of $\theta_{j,k}$, so the experimental results conclusively demonstrate our analytical characterization in the qubit and qutrit situation. Points out of the regions mainly attributed to statistical fluctuation and instability of the interferometer, see Supplemental Material \cite{Supplementary} for a detailed analysis.

\textit{Conclusion.---} In this work, we develop the notion of uncertainty region to geometrically characterize the entire area of variances of the observables. This notion gives an entire description of the uncertainty relations of a pair of or more incompatible observables. Once the UR has been characterized, any other forms of variance-based uncertainty relations can also be determined. URs are also viewed as a kind of state-independent uncertainty relation, because it gives region for all quantum state (pure or mixed). Compare to the previous works, we point out that UR is necessarily semialgebraic, both for pure states and mixed states. Further, by applying the Jordan's Lemma, we give a complete analytical characterization of URs for a pair of qubit observables and a pair of arbitrary dimensional projective observables. Versatile photonic setups are used for observing URs of qubit and qutrit projectors and the results are highly consistent with the theory. Our work provides a new method for studying the uncertainty regions of incompatible observables and gives a new perspective for understanding and representing the uncertainty of quantum observables.

\begin{acknowledgments}
This work was supported by the National Key Research and Development Program of China (grant no. 2017YFA0303703), the National Natural Science Foundation of China (grant nos. 11690032, 61490711, 11474159 and 11574145) and ARC DECRA 180100156.
\end{acknowledgments}

\clearpage
\begin{widetext}
\section{Supplemental Materials for "Entirety of Quantum Uncertainty and Its Experimental Verification"}
\section{Detailed Derivation}
\subsection{Proof of Theorem 2}
\par
Let $A=P+iQ$.
According to the famous Hausdorff-Toeplitz theorem~\cite{Hausdorff1919}, we have
that the numerical range of $A$
$$W(A)=\{\langle \psi|A|\psi\rangle:|\psi\rangle\mathrm{\ is\ pure}\}$$
is convex and compact.
In other words,
it equals
$$W'(A)=\{\Tr(\rho A):\rho~\mathrm{ranges~over~all~mixed~states.}\}$$
Directly, we know that $R_m=R_p$ in this case.

However, when the number of projectors in the set increase, this is not true generally. For example, considering the following four qubit states
\begin{equation*}\label{SIC_state}
  \ket{\psi_1}=\ket{0},\ket{\psi_2}=-\frac{1}{\sqrt{3}}\ket{0}+\sqrt{\frac{2}{3}}\ket{1},\\
  \ket{\psi_3}=-\frac{1}{\sqrt{3}}\ket{0}+e^{i\frac{2\pi}{3}}\sqrt{\frac{2}{3}}\ket{1},\\
  \ket{\psi_4}=-\frac{1}{\sqrt{3}}\ket{0}+e^{-i\frac{2\pi}{3}}\sqrt{\frac{2}{3}}\ket{1}
\end{equation*}
which form a regular tetrahedron on the Bloch ball. Uncertainty region of their corresponding projectors $(\Delta P_{\ket{\psi_1}},$ $\Delta P_{\ket{\psi_2}},\Delta P_{\ket{\psi_3}},\Delta P_{\ket{\psi_4}})$ no longer have the statement $R_m=R_p$. Considering the maximally mixed state $\rho=\frac 1 2 I$, it generate the point $(1/4,1/4,1/4,1/4)$ in $R_m$, while for arbitrary pure qubit states, no one can achieve this point. Because for all pure states, it must have that $\Delta P_{\ket{\psi_1}}+\Delta P_{\ket{\psi_2}}+\Delta P_{\ket{\psi_3}}+\Delta P_{\ket{\psi_4}})=\frac 2 3$.

\subsection{Proof of Theorem 3}
\par
\begin{proof}
Let us consider the one-dimensional Gaussian wave packet at time 0:
$$
\psi(x,0) = \frac{1}{\sqrt{a\sqrt{\pi}}}e^{-\frac{x^2}{2a^2}+ik_0\frac{x}{a}}.
$$
It is straightforward to check $\psi$ is normalized.
Solving the Schr\"odinger equation
$$
i\hbar\frac{\partial}{\partial t}\psi(x,t) = \frac{p^2}{2m}\psi(x,t) = -\frac{\hbar^2}{2m}\frac{\partial^2}{\partial x^2}\psi(x,t)
$$
with the initial condition $\psi(x,0)$ results that:
$$
\psi(x,t) = \frac{1}{\sqrt{a\sqrt{\pi}}}\frac{1}{\sqrt{1+\frac{i\hbar}{ma^2}t }}e^{-\frac{1}{2}k_0^2}e^{-\frac{1}{1+\frac{i\hbar}{ma^2}t}\frac{1}{2}\left(\frac{x}{a}-ik_0\right)^2}.
$$
After tedious calculations, we obtain the expectations of $x,x^2,p,p^2$:
\begin{align*}
\<x\>_{\psi(x,t)} &= \int_{-\infty}^\infty \psi(x,0)^\ast\psi(x,0)x\diff x = \frac{\hbar k_0t}{ma},\\
\<x^2\>_{\psi(x,t)} &= \int_{-\infty}^\infty \psi(x,0)^\ast\psi(x,0)x^2\diff x \\
&= \frac{a^2}{2}+\frac{\hbar^2 t^2}{2m^2a^2}+\frac{\hbar^2 k_0^2t^2}{m^2a^2},\\
\<p\>_{\psi(x,t)} &= -i\hbar\int_{-\infty}^\infty \psi(x,0)^\ast\frac{\partial}{\partial x}\psi(x,0)\diff x = \hbar \frac{k_0}{a},\\
\<p^2\>_{\psi(x,t)} &= -\hbar^2\int_{-\infty}^\infty \psi(x,0)^\ast\frac{\partial^2}{\partial x^2}\psi(x,0)\diff x \\
&= \hbar^2\frac{1}{2a^2}+ \hbar^2\frac{k_0^2}{a^2}
\end{align*}
which leads to:
\begin{align*}
&\Delta x_{\psi(x,t)} = \sqrt{\frac{a^2}{2}+\frac{\hbar^2 t^2}{2m^2a^2}}, \\
&\Delta p_{\psi(x,t)} = \sqrt{\hbar^2\frac{1}{2a^2}}.
\end{align*}
With proper choices of $a$ and $t\ge0$, $(\Delta x_{\psi(x,t)},\Delta p_{\psi(x,t)})$ fulfills all regions of $S = \{(x,y):xy\geq \frac{\hbar}{2},x,y>0\}$, or
equivalently, $R_p = S$.
Moreover, $R_p\subseteq R_m \subseteq S$ which implies $R_p=R_m=S$.
\end{proof}
\subsection{Derivation of $R^{(2)}$}
\par
In the Bloch sphere representation, an arbitrary qubit density matrix $\rho$ and two projectors $A$ and $B$ can be written as
	\begin{align*}
	\rho=\frac{1}{2}(I + \vec{r}\cdot\vec{\sigma})
	\end{align*}
    \begin{align*}
	A=\frac{1}{2}(I + \vec{a}\cdot\vec{\sigma})
	\end{align*}
	\begin{align*}
	B=\frac{1}{2}(I + \vec{b}\cdot\vec{\sigma})
	\end{align*}
where $\vec{r}=(r_x,r_y,r_z),\vec{a}=(a_x,a_y,a_z),\vec{b}=(b_x,b_y,b_z)$ are unit Bloch vectors and $\vec{\sigma}=(\sigma_x,\sigma_y,\sigma_z)$ is the trceless Pauli vector. According to Lemma 1 in the main text, $\vec{a},\vec{b}$ can be written in the basis that $\vec{a}=(0,0,1)$ and $\vec{b}=(\sin2\theta,0,\cos2\theta)$. Then by using the identity $(\vec{r}\cdot\vec{\sigma})(\vec{a}\cdot\vec{\sigma})=(\vec{r}\cdot\vec{a})I+i(\vec{r}\times\vec{a})\cdot\vec{\sigma}$ and according to Born's rule, it can be derived that
	\begin{align*}
	\Delta A=\frac{1}{4}\Big(1-(\vec{r}\cdot\vec{a})^2\Big)
	\end{align*}
	\begin{align*}
	\Delta B=\frac{1}{4}\Big(1-(\vec{r}\cdot\vec{b})^2\Big)
	\end{align*}
	which can be further written as
	\begin{equation}\label{eqn:r_x}
	r_x\sin{2\theta}=\pm\sqrt{1-4\Delta B}-r_z\cos{2\theta}
	\end{equation}
	\begin{equation}\label{eqn:r_z}
		r_z^2=1-4\Delta A
	\end{equation}
The above two equations characterize the relationship between $\Delta A$ and $\Delta B$ for different Bloch vector $\vec{r}$. By ranging $\vec{r}$ over the whole Bloch ball, the set of $R^{(2)}_m$ can be fully characterized. From Eq.(\ref{eqn:r_x}) and Eq.(\ref{eqn:r_z}), it can be seen that the only constraint on $\vec r$ is $r_x^2+r_z^2\le1$, that is the positivity constraint. Embedding Eq.(\ref{eqn:r_x}) and Eq.(\ref{eqn:r_z}) into the constraint, $R^{(2)}$ can be derived from the following inequality
	\begin{equation}
    \label{rx_rz_constraint}
	\pm2\cos2\theta\sqrt{(1-4\Delta A)(1-4\Delta B)}\le4(\Delta A+\Delta B)-(1+\cos^22\theta),
	\end{equation}\label{reg}
	which is equivalent to the union set of the following two regions
	\begin{align}
              \label{seqn:R_1}
	R_1^{(2)}:
	\left\{
	\begin{array}{lr}
	\Delta A+\Delta B\geq (1+\cos^22\theta)/4,&\\
	\Delta A,\Delta B\leq 1/4,
	\end{array}
	\right.
	\end{align}
	\begin{align}
       \label{seqn:R_2}
	R_2^{(2)}:
	\left\{
	\begin{array}{lr}
	\Delta A+\Delta B\leq(1+\cos^22\theta)/4,&\\
	64x_1^2/(1+\cos4\theta)+64y_1^2/(1-\cos4\theta)\leq1,
	\end{array}
	\right.
	\end{align}
where $x_1$ and $y_1$ are
	\begin{align*}
		\left( \begin{array}{ccc}
		x_1 \\
		y_1
		\end{array}
		\right)=
		\left(
		\begin{array}{ccc}
		\cos{\frac{\pi}{4}} & \sin{\frac{\pi}{4}} \\
		-\sin{\frac{\pi}{4}} & \cos{\frac{\pi}{4}}
		\end{array}
		\right)
		\left(
		\begin{array}{ccc}
		\Delta A-{1}/{8} \\
		\Delta B-{1}/{8}
		\end{array}
		\right).
	\end{align*}

\subsection{Derivation of $R^{(d)}(d\geq3)$}
As shown in the main text, by applying Lemma 1, all the calculations in higher dimensional system can be reduced into the qubit subspace. The two rank 1 projectors $A$ and $B$ can be written as
	\begin{align*}
	A=\begin{pmatrix}
	1 & 0 & \quad \\
	0 & 0 & \quad \\
	\quad & \quad & \ddots
	\end{pmatrix},
    B=\begin{pmatrix}\cos^2\theta & \cos\theta\sin\theta & \quad \\ \cos\theta\sin\theta & \sin^2\theta & \quad \\ \quad & \quad & \ddots
	\end{pmatrix},
	\end{align*}
in which the non-zero $2\times2$ block indicating the qubit subspace spaned by the two projectors. Thus the density matrix $\rho$ can also be written in the same qubit subspace
	\begin{align*}
	\rho = \begin{pmatrix} \rho^{(2)} \\
	& & \ddots
	\end{pmatrix},
	\end{align*}
where all the blank space set to 0 and $\rho^{(2)}$ is an unnormalized qubit density matrix
	\begin{align*}
	\rho^{(2)}=\frac{1}{2}(\alpha I+\vec{r}\cdot\vec{\sigma}),r^2=1
	\end{align*}
satisfying $0\leq\alpha\leq1$ and $r^2\leq\alpha^2$ as the positivity constraint. Here $\alpha=\Tr (\rho^{(2)})$ is the probability that $\rho$ been projected into the qubit subspace. With the above representation, all the calculation can be reduced in the qubit subspace and it is easy to derive
\begin{eqnarray}\label{delta_A_delta_B_d_3}
  \nonumber\Delta A &=& \frac{1}{2}(\alpha+\vec{r}\cdot\vec{a})-\frac{1}{4}(\alpha+\vec{r}\cdot\vec{a})^2 \\
  \Delta B &=& \frac{1}{2}(\alpha+\vec{r}\cdot\vec{b})-\frac{1}{4}(\alpha+\vec{r}\cdot\vec{b})^2.
\end{eqnarray}
 Then we can solve the coordinates of $\vec{r}$ from Eq. (\ref{delta_A_delta_B_d_3})
 \begin{eqnarray}
 \label{rz_rx_d_3}
   \nonumber r_z &=& 1-\alpha\pm\sqrt{1-4\Delta A} \\
   r_x &=& \frac{1}{\sin{2\theta}}[(1-\alpha-r_z\cos{2\theta})\pm\sqrt{1-4\Delta B}].
 \end{eqnarray}
 Similarly, put the above equations into the positivity constraint $r_x^2+r_z^2\leq\alpha^2$, we can arrive a quadratic inequality about $\alpha$,
 \begin{equation}\label{quadratic inequality}
   f_1(\Delta A,\Delta B,\theta)\alpha^2+f_2(\Delta A,\Delta B,\theta)\alpha+f_3(\Delta A,\Delta B,\theta)\leq0,
 \end{equation}
 where $f_1(\Delta A,\Delta B,\theta)$, $f_2(\Delta A,\Delta B,\theta)$ and $f_3(\Delta A,\Delta B,\theta)$ are functions of $\Delta A$, $\Delta B$ and $\theta$. Given the fact that $0\leq\alpha\leq1$, if there exists an $\alpha\in[0,1]$ and corresponding parameters $(\Delta A,\Delta B,\theta)$ that Eq. (\ref{quadratic inequality}) can be satisfied, then the point $(\Delta A, \Delta B)$ is attainable, thus by discussing the above existing problem about Eq. (\ref{quadratic inequality}) for all the combinations of signs in Eq. (\ref{rz_rx_d_3}), $R^{(d)}$ can be derived. Specifically, for example, when $r_x$ and $r_z$ in Eq. (\ref{rz_rx_d_3}) both take minus sign, we have
 \begin{eqnarray}
   \nonumber f_1 &=& \tan^2{\theta}, \\
   \nonumber f_2 &=& \sec^2{\theta}(\sqrt{1-4\Delta A}+\sqrt{1-4\Delta B}-2), \\
   \nonumber f_3 &=& -\frac{1}{\sin^2{2\theta}}[2\cos{2\theta}(\sqrt{1-4\Delta A}-1)(\sqrt{1-4\Delta B}-1)-(\sqrt{1-4\Delta A}-1)^2-(\sqrt{1-4\Delta B}-1)^2].
 \end{eqnarray}
By defining the quadratic function $F(\alpha)=f_1\alpha^2+f_2\alpha+f_3$, which is symmetry with respect to the line $\alpha=\alpha_0=-\frac{f_2}{2f_1}$ and get its minimum value at $(\alpha_0,\frac{4f_1f_3-f_2^2}{4f_1})$, the above existing problem can be satisfied by the following two situations,
 \begin{equation*}
   \left\{
	\begin{array}{lr}
	\alpha_0\geq1,\\
    F(1)\leq0
	\end{array}
	\right.
 \end{equation*}
 \begin{equation*}
   \left\{
	\begin{array}{lr}
	0\leq\alpha_0\leq1,\\
    \frac{4f_1f_3-f_2^2}{4f_1}\leq0
	\end{array}
	\right.
 \end{equation*}
 The above two situations correspond to the two graphs in Fig. \ref{fig:quadratic_curve}. Put the explicit form of $f_1,f_2,f_3$ into the above equation, we get the following two attainable regions of $(\Delta A,\Delta B)$
\begin{align}
              \label{seqn:R_1_d3}
	R_1^{(d)}:
	\left\{
	\begin{array}{lr}
	2-\sqrt{1-4\Delta A}-\sqrt{1-4\Delta B}\geq2\sin^2{\theta},&\\
    \quad  \\
    \pm2\cos2\theta\sqrt{(1-4\Delta A)(1-4\Delta B)} \\
    -4(\Delta A+\Delta B)+(1+\cos^22\theta)\leq0,
	\end{array}
	\right.
	\end{align}
	\begin{align}
       \label{seqn:R_2_d3}
	R_2^{(d)}: 0\leq2-\sqrt{1-4\Delta A}-\sqrt{1-4\Delta B}<2\sin^2{\theta}.
\end{align}
Next, we consider the other 3 combinations of signs in Eq. (\ref{rz_rx_d_3}) and the union set of all the attainable regions equal to the uncertainty region $R^{(d)}$. However, it can be verified that the attainable regions in Eq. (\ref{seqn:R_1_d3}) cover the other three combinations and equal to $R^{(d)}$. So for higher dimensional quantum system their uncertainty regions of two projecters $A,B$ are given by Eq. (\ref{seqn:R_1_d3}) and Eq. (\ref{seqn:R_2_d3}).

\begin{figure*}
  \centering
  \addtocounter{figure}{1}
  \subfigure[]{
    \label{fig:quadratic_1} 
    \includegraphics[width=0.38\textwidth]{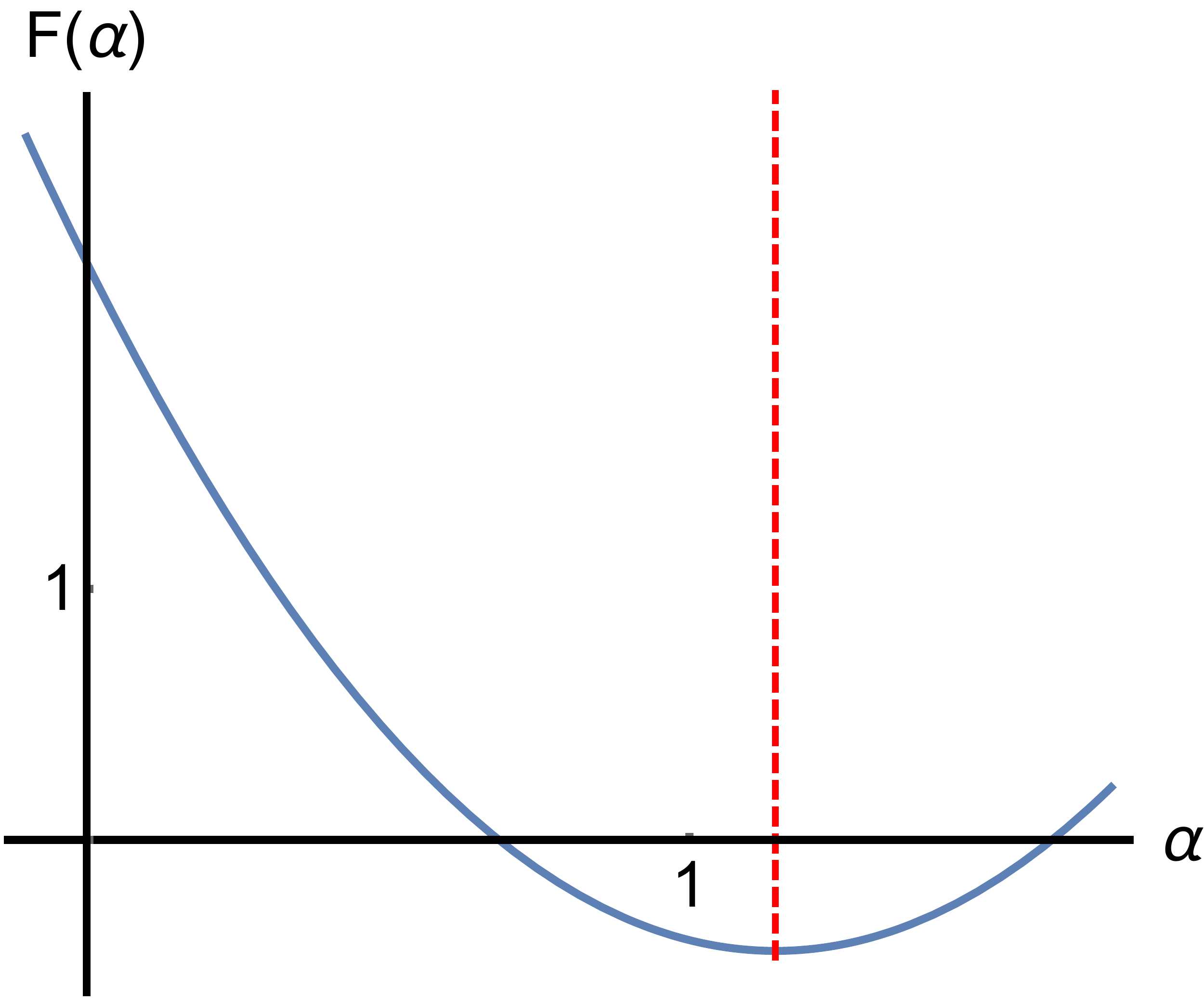}}\quad\quad\quad\quad\quad\quad\quad\quad\quad
  \subfigure[]{
    \label{fig:quadratic_0_1} 
    \includegraphics[width=0.38\textwidth]{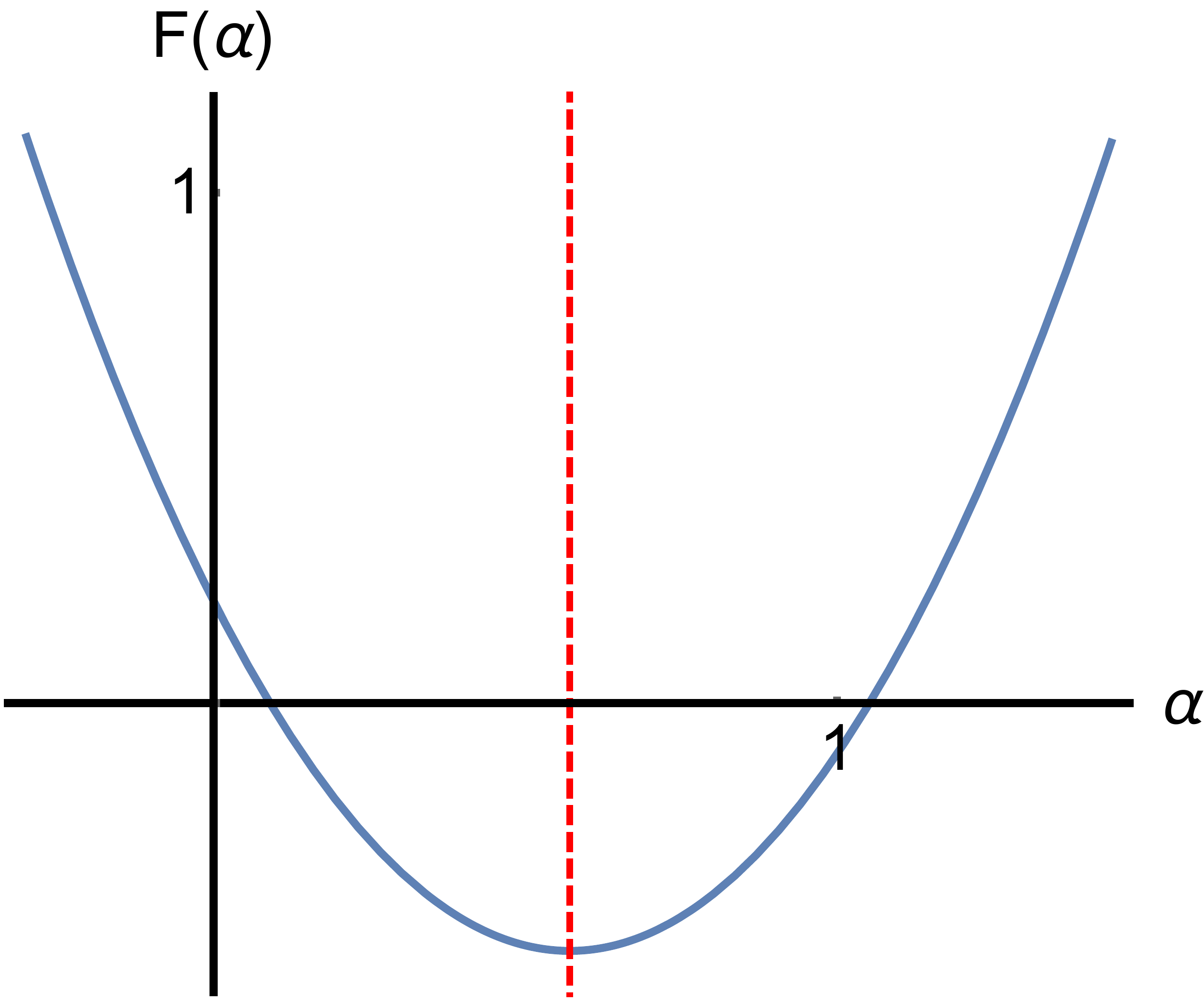}}
  \addtocounter{figure}{-1}
  \caption{Two exemplary curves (bule lines) and its symmetrical axes (red dashed lines) of $F(\alpha)$ satisfying the existing problem and therefore satisfying the positivity constraint, . (a) When $\alpha_0\geq1$, as long as $F(1)<0$, Eq. (\ref{quadratic inequality}) can be satisfied. (b) When $\alpha_0\in[0,1)$, as long as the minimum value of $F(\alpha)$ less than $0$, Eq. (\ref{quadratic inequality}) can be satisfied by some $(\Delta A,\Delta B)$ for a given $\theta$.}
  \label{fig:quadratic_curve} 
\end{figure*}
\section{Experimental Details}
\subsection{State Preparation}
We use BD, wave plates, a LCPR and a QWP-HWP-QWP configuration to prepare arbitrary pure qutrit states in the form
\begin{equation*}
  \ket{\psi}=e^{i\phi_1}\cos{\theta_A}\ket{0}+e^{i\phi_2}\sin{\theta_A}\sin{\theta_B}\ket{1}-\sin\theta_A\cos{\theta_B}\ket{2},
\end{equation*}
encoded in the polarization and spatial optical modes. The two phase factor $\phi_1$ and $\phi_2$ are implemented by the LCPR and the QWP-HWP-QWP configuration, respectively. The LCPR (Thorlabs, LCC1113-B) is set with its optical axis parallel to the horizontal polarization and add a relative phase between horizontal and vertical polarization. Due to the tiny separating distance (4mm) of the two spatial modes and limited retardance uniformity of the liquid crystal phase retarder, we use a QWP-HWP-QWP configuration as the second phase retarder with two QWP setting at $45^\circ$ and an electronically-controlled HWP setting at diferent angles to realize different $\phi_2$~\cite{Xie2019}.

\subsection{Measurement of Uncertainty Regions}
As depicted in Fig. 2(c) in the main text, by setting a appropriate angle $\theta_2$ of HWP2, we implement three orthogonal qutrit projective measurements on the output ports of the measurement module. Therefore, HWP2 together with HWP3 perform the following unitary in the Jones matrix notation on the input state
	\begin{align*}
    \label{U}
	U = \begin{pmatrix} \cos2\theta_2 & \sin2\theta_2 & 0 \\
	\sin2\theta_2 & -\cos2\theta_2 & 0 \\
    0 & 0 & 1
	\end{pmatrix},
	\end{align*}
which can be written in a more intuitive form
\begin{equation}
\label{U_waiji}
  U=\ket{0}\bra{\psi_0}+\ket{1}\bra{\psi_1}+\ket{2}\bra{\psi_2},
\end{equation}
where $\ket{\psi_0}=\cos2\theta_2\ket{0} + \sin2\theta_2\ket{1},\ket{\psi_1}=\sin2\theta_2\ket{0}+\cos2\theta_2\ket{1}$ and $\ket{\psi_2}=\ket{2}$. From Eq. (\ref{U_waiji}), it can be seen that $U$ transform the input state from the basis of $\ket{\psi_i}(i=0,1,2)$ into experimental basis, then after BD3, detector $\rm{D_0}$ equivalently perform the projective measurement $\ket{\psi_0}\bra{\psi_0}$. By setting different angles $\theta_{2j}=\{0,\frac{\pi}{72},\frac{\pi}{12},\frac{\pi}{8}\}$, the corresponding projective measurements in dector $\rm{D_0}$ can be written as
\begin{equation*}
  P_j=\begin{pmatrix} \cos^22\theta_{2j} & \cos2\theta_{2j}\sin2\theta_{2j} & 0 \\
	\cos2\theta_{2j}\sin2\theta_{2j} & \sin^22\theta_{2j} & 0 \\
    0 & 0 & 0
	\end{pmatrix},
\end{equation*}
thus the angle between $P_j$ and $\ket{0}\bra{0}$ is $2\theta_j$ and for different pair of $\{P_j,P_k\}$, their angle is $\theta_{j,k}=2(\theta_k-\theta_j)$.

The above discussion showed how we implement different qutrit projective measurements. As for the qubit projectors, we only need to post-select and re-normalize the measured statistics of dector $\rm{D_0}$ and $\rm{D_1}$. Then different qubit projectors $p_j$ is realized in dector $\rm{D_0}$, where $p_j$ is the non-zero 2 by 2 block in $P_j$. By this method, we have measured URs of qubit and qutrit projective measurements simultaneously (in one experiment).

\subsection{Experimental Error Analysis}
As mentioned in the main text, points generated by all states should be in its corresponding UR. Therefore, imperfection in the state preparation stage may lead to derivation between the theoretical value but won't lead to experimental obtained points falling out of the theoretical region. As show in our experiment results, there exist several points out of the theoretical region in the fourth figure in Fig 3(a) in the main text. This may attributed to imperfection and inaccuracy of the wave-plates in the measurement stage or instability of the interferometer or simply the statistical fluctuation. Considering the tiny retardation errors (typically $\sim\lambda/300$ where $\lambda=830$nm), inaccurate setting angle (typically  $\sim0.2$ degree) and misalignment (typically  $\sim0.1$ degree) of wave plates and the high photon counting (about 45000/s), the experimental error mainly caused by the slowly drift and slight vibrating of the interferometer during the measurement progress (about 80 minutes for all 400 states), which lead to a limited interference visibility and thus biased projective measurement. However, during the whole measuring progress of $(\Delta p_3,\Delta p_4)$, the average interference visibility is above $98\%$. Except for the before mentioned several points in the fourth figure in Fig 3(a), almost all experimental datas are located in its corresponding theoretical regions.
\end{widetext}

\end{document}